\documentclass[11pt,fleqn,a4paper]{article}

\usepackage{amsmath,amssymb,amsthm,enumerate}
\newcommand{\noprint}[1]{}

\allowdisplaybreaks

\newcommand{\todo}[1][\null]{\ensuremath{\clubsuit}}

\newcommand{\p}{\partial}
\newcommand{\const}{\mathop{\rm const}\nolimits}

\newtheorem{theorem}{Theorem}

\newtheorem{corollary}{Corollary}
\newtheorem{proposition}{Proposition}
{\theoremstyle{definition}

}

\begin{document}

\par\noindent {\LARGE\bf
Reduction operators of Burgers equation
\par}

{\vspace{4mm}\par\noindent {\large Oleksandr A.~Pocheketa$^\dag$ and~Roman O.~Popovych$^\ddag$}
\par\vspace{2mm}\par}

{\vspace{2mm}\par\noindent \it
$^\dag{}{}^\ddag$Institute of Mathematics of NAS of Ukraine, 3 Tereshchenkivska Str., 01601 Kyiv, Ukraine\par}

{\vspace{2mm}\par\noindent {\it
$^\ddag\phantom{^\dag}$Faculty of Mathematics, University of Vienna, Nordbergstra{\ss}e 15, A-1090 Vienna, Austria
}\par}
{\vspace{2mm}\noindent {E-mail: \it
$^\dag$pocheketa@yandex.ua, $^\ddag$rop@imath.kiev.ua }\par}

{\vspace{7mm}\par\noindent\hspace*{8mm}\parbox{144mm}{\small
The solution of the problem on reduction operators and nonclassical reductions of the Burgers equation
is systematically treated and completed.
A~new proof of the theorem on the special ``no-go'' case of regular reduction operators is presented,
and the representation of the coefficients of operators
in terms of solutions of the initial equation is constructed for this case.
All possible nonclassical reductions of the Burgers equation to single ordinary differential equations
are exhaustively described.
Any Lie reduction of the Burgers equation proves to be equivalent via the Hopf--Cole transformation
to a parameterized family of Lie reductions of the linear heat equation.
}\par\vspace{2mm}}

\noprint{
Keywords: Burgers equation, nonclassical symmetry, reduction operator, Lie symmetry, exact solution
}

\section{Introduction}

The second-order evolution equation
\begin{equation}\label{BE}
L[u]:=u_t+uu_x+u_{xx}=0
\end{equation}
was proposed by Burgers~\cite{burg39a,burg48a} as a one-dimensional turbulence model.
The equation~\eqref{BE} is also applied to model other phenomena
in physics, chemistry, mathematical biology, etc.
A fairly complete review of properties of the Burgers equation can be found in~\cite[Chapter~4]{whit74book}.

It is well known that the equation~\eqref{BE} is linearized to the heat equation~$v_t+v_{xx}=0$
using the so-called Hopf--Cole transformation $u=2v_x/v$~\cite[p.~102]{fors1906book}.
At the same time, due to the importance of the Burgers equation for various applications,
the exhaustive study of its properties in the framework of symmetry analysis is still topical.

Lie symmetries of the Burgers equation and~some of its generalizations were studied since the 1960s.
The maximal Lie invariance algebra~$\mathfrak g^{\rm B}$
of the equation~\eqref{BE} was first computed by Katkov~\cite{katk65a} in the course of
group classification of differential equations of the general form
$u_t+uu_x=(f(u)u_x)_x$.
The algebra~$\mathfrak g^{\rm B}$ is spanned by the vector fields
\begin{gather*}
\mathcal P_t=\partial_t,\quad
\mathcal D=2t\partial_t+x\partial_x-u\partial_u,\quad
\mathcal K=t^2\partial_t+tx\partial_x+(x-ut)\partial_u,\\
\mathcal P_x=\partial_x,\quad
\mathcal G=t\partial_x+\partial_u.
\end{gather*}
The complete point symmetry group~$G^{\rm B}$ of the equation~\eqref{BE} consists of the transformations
\begin{equation*}
\widetilde{t}=\frac{\alpha t+\beta}{\gamma t+\delta},\quad\widetilde{x}=\frac{\kappa x+\mu_1
t +\mu_0}{\gamma t+\delta},\quad\widetilde u=\frac{\kappa(\gamma t+\delta)u-\kappa\gamma
x+\mu_1\delta-\mu_0\gamma}{\alpha\delta-\beta\gamma},
\end{equation*}
where $(\alpha,\beta,\gamma,\delta,\kappa,\mu_0,\mu_1)$ is
an arbitrary set of constants defined up to a nonzero multiplier, and
$\alpha\delta-\beta\gamma=\kappa^2>0$.
Up to composition with continuous point symmetries,
the group~$G^{\rm B}$ contains the single discrete symmetry $(t,x,u)\rightarrow(t,-x,-u)$.

Generally, reductions of partial differential equations using their Lie symmetries
do not provide sufficiently large families of exact solutions of these equations.
The nonclassical method of reduction was proposed in~\cite{blum67th} (see also~\cite{blum69a}) in order to utilize
a wider class of vector fields than Lie symmetries.
Later such vector fields were called nonclassical symmetries~\cite{LeviWinternitz1989}
or conditional symmetries~\cite{fush93book,FushchychZhdanov92,zhda99a} or reduction operators~\cite{popo08b}.
The notion of nonclassical symmetries can be extended in several directions,
e.g., to the concept of weak symmetry introduced in~\cite{olve87a},
which is also closely related to compatibility theory of differential equations
and the general method of differential constraints~\cite{olve86a,olve94a,pucc92b,Sidorov&Shapeev&Yanenko1984,yane64a}.
Generalized notions of ansatzes and reductions associated with weak symmetries were intensively discussed,
e.g., in \cite{burd96a,olve87a,pucc00a,sacc04b}, see also references therein.

A \emph{reduction operator} of the equation~\eqref{BE} is a vector field of the general form
\begin{equation}\label{RO}
Q=\tau(t,x,u)\partial_t+\xi(t,x,u)\partial_x+\eta(t,x,u)\partial_u,
\end{equation}
where the coefficients~$\tau$ and~$\xi$ do not simultaneously vanish,
that allows one to construct an ansatz reducing the initial equation~\eqref{BE} to an ordinary differential equation.
See, e.g., \cite{zhda99a} for the general definition of involutive families of reduction operators.
Every Lie symmetry operator is a reduction operator.
The multiplication by nonvanishing functions of~$(t,x,u)$ generates an equivalence relation
on the set of reduction operators.
The determining equations on coefficients of a reduction operator~$Q$ are derived
from the conditional invariance criterion~\cite{FushchychZhdanov92,olve96book,zhda99a}
\begin{equation}\label{EqCondInvCriterion}
Q_{(2)}L[u]\!\mid_{\mathcal{L}\cap\mathcal{Q}^{(2)}}=0.
\end{equation}
Here $Q_{(2)}$ is the second prolongation of the vector field~$Q$,
$\mathcal{L}$ is the manifold in the second-order jet space $J^{(2)}$ that corresponds to the Burgers equation $L[u]=0$,
and $\mathcal{Q}^{(2)}$ is the manifold in the same jet space
determined by the invariant surface condition $Q[u]=0$
jointly with its differential consequences $D_tQ[u]=0$ and~$D_xQ[u]=0$,
$Q[u]=\eta-u_t-\xi u_x$ is the characteristic of the vector field~$Q$,
$D_t$ and~$D_x$ are the operators of total differentiation with respect to $t$ and $x$, respectively.
In view of the evolution kind of the Burgers equation it is natural to partition
the set of its reduction operators into two subsets, singular and regular,
depending on whether or not the coefficient~$\tau$ vanishes~\cite{kunz08b}.
Up to the above equivalence relation, one can assume $(\tau,\xi)=(0,1)$ and $\tau=1$
for singular and regular reduction operators of the Burgers equation, respectively.

It is the Burgers equation that was first considered from the nonclassical symmetry point of view
after the prominent paper~\cite{blum69a}.
Namely, in~\cite{wood71th,wood71a} the determining equations for regular reduction operators of~\eqref{BE} were
derived under the gauge $\tau=1$
and a few of their particular solutions satisfying the additional constraint $\xi_u=0$ were constructed.
The corresponding results are available in~\cite{ames72book}.
The determining equations for both regular and singular nonclassical symmetries of~\eqref{BE}
were presented in~\cite{pucc92a}. Therein the regular case was studied in detail under the gauge $\tau=1$,
for which the consideration was shown to be partitioned into three cases, $\xi_u=0$, $\xi_u=1$ and $\xi_u=-\frac12$.
The case $\xi_u=0$ proved to result merely in nonclassical symmetries which are equivalent to Lie symmetries.
(Within the framework of the direct method, the same result was earlier obtained in~\cite{clar89a}
in terms of the corresponding ansatzes and reductions.)
The unique reduction operator $\partial_t+u\partial_x$ satisfying the constraint $\xi_u=1$ was also found
and used for reducing the Burgers equation.
For the case $\xi_u=-\frac12$ some particular solutions of the determining equations jointly
with the corresponding ansatzes and invariant solutions of the Burgers equation were constructed.
The above consideration of regular nonclassical symmetries from~\cite{pucc92a} was extended in~\cite{arri93a}
with more particular solutions satisfying the constraint $\xi_u=-\frac12$.
Still wider families of particular solutions of the determining equations in this case
were given in~\cite{cher98e,olve96book}.
In~\cite{clar94a} an algorithmic procedure to derive determining equations for nonclassical symmetries was proposed,
and the Burgers equation was one of the illustrative examples for application of this procedure.

The system~$\mathcal S_{\rm B}$ of determining equations for the case~$\xi_u=-\frac12$
was not well investigated for a surprisingly long time
although the study of the analogous system~$\mathcal S_{\rm h}$
for regular reduction operators of the linear heat equation $v_t+v_{xx}=0$,
whose form is very similar to~$\mathcal S_{\rm B}$, had already been completed in~\cite{fush93book,fush92a}.
See also \cite{webb90a} for preliminary results on~$\mathcal S_{\rm h}$
and \cite{fush94a,popo95a,popo08b} for further generalizations
to (1+1)-dimensional second-order linear evolution equations.
The system~$\mathcal S_{\rm B}$ was first linearized in~\cite{mans99a} in a fashion similar to~\cite{fush92a}.
Namely, this system was reduced by a differential substitution
to the uncoupled system of three copies of the linear heat equation.
As shown in~\cite{arri02a}, the systems~$\mathcal S_{\rm B}$ and~$\mathcal S_{\rm h}$
as well as the substitutions linearizing them
can be interpreted in terms of the matrix Burgers equation and the matrix Hopf--Cole transformation.

The above in fact means that the case~$\xi_u=-\frac12$ can be referred~to as ``no-go''.
In general, looking for reduction operators in a family of vector fields is said to result in a ``no-go'' case
if the corresponding system of determining equations for coefficients of reduction operators
is reduced to a well-determined system whose general solution cannot be represented in a closed form,
and, moreover, solving this system is equivalent, in a certain sense, to solving the initial equation.

Singular reduction operators of the Burgers equation were in fact not studied until \cite{Popovych1998,zhda98a},
where no-go results of~\cite{fush92a} on reduction operators with $\tau=0$ for the linear heat equation
were extended to general evolution equations of order greater than one.
These results were treated in~\cite{kunz08b} within the framework of singular reduction operators.

In this paper we intend to enhance and complete the above results on nonclassical symmetries and reductions
of the Burgers equation.
In particular, extending methods from~\cite{popo08b} we present a new proof on the linearization
of the system~$\mathcal S_{\rm B}$
and show that solutions of this system are expressed via triples of solutions of the Burgers equation.
We first exhaustively describe all possible nonclassical reductions of the Burgers equation
to single ordinary differential equations
including reductions associated with the no-go case $\tau=1$ and $\xi_u=-\frac12$.
A part of the description is the assertion stating that any Lie reduction of the Burgers equation
is equivalent via the Hopf--Cole transformation
to a parameterized family of Lie reductions of the linear heat equation.

\section{Singular reduction operators}\label{NCSdeteqSECTION}

As the Burgers equation~\eqref{BE} is a (1+1)-dimensional second-order evolution equation,
every its reduction operator of the form~\eqref{RO} with $\tau=0$ is singular
since the corresponding reduced equation is of a lower (namely, the first) order~\cite{kunz08b}.
All basic results on such reduction operators of~\eqref{BE} follow from
the general results on singular reduction operators of co-order singularity one.
After setting $\xi=1$ in~$Q$ due to the equivalence of reduction operators,
the conditional invariance criterion~\eqref{EqCondInvCriterion} implies
a single determining equation on the single coefficient~$\eta=\eta(t,x,u)$,
\begin{equation}\label{EqDetSingularCase}
\eta_t+u\eta_x+\eta^2+\eta_{xx}+2\eta\eta_{xu}+\eta^2\eta_{uu}=0.
\end{equation}
The equation~\eqref{EqDetSingularCase} can be transformed to~\eqref{BE} \cite{kunz08b,zhda98a}.
Namely, the composition of the differential substitution $\eta=-\Phi_x/\Phi_u$ with $\Phi_u\ne0$,
where $\Phi$ is a smooth function of $(t,x,u)$,
and the hodograph transformation, where the new independent variables are $\tilde t=t$, $\tilde x=x$ and $\varkappa=\Phi$
and the new dependent variable is $\tilde u=u$, reduces the equation~\eqref{EqDetSingularCase} to
the initial equation~\eqref{BE} on the function $\tilde u=\tilde u(\tilde t,\tilde x,\varkappa)$
with $\varkappa$ playing the role of a parameter.
Moreover, up to equivalences of reduction operators and solution families,
there exists a bijection between one-parameter families of solutions of the equation~\eqref{BE}
and its reduction operators with zero coefficients of $\p_t$.
Namely, each operator of the above kind corresponds to
the family of solutions which are invariant with respect to this operator.
The problems of construction of all one-parameter solution families of~\eqref{BE}
and the exhaustive description of its reduction operators with zero coefficients of $\p_t$
are completely equivalent.
Given a family $\mathcal F=\{u=f(t,x,\varkappa)\}$ of solutions of~\eqref{BE}
parameterized by a single essential parameter~$\varkappa$,
the corresponding singular reduction operator is $Q=\p_x-(\Phi_x/\Phi_u)\p_u$,
where the function $\Phi$ is obtained by solving the equality $u=f(t,x,\varkappa)$
with respect to~$\varkappa$, $\varkappa=\Phi(t,x,u)$.
The ansatz $u=f(t,x,\varphi(\omega))$, where $\omega=t$, associated with~$Q$,
reduces the equation~\eqref{BE}  to the equation $\varphi_\omega=0$.
The simplicity of the reduced equation is explained by the specific choice
of the ansatz based on knowing the one-parameter family~$\mathcal F$ of solutions.

\section{Regular reduction operators}

Now we look for regular reduction operators of the Burgers equation~\eqref{BE},
which are of the form~\eqref{RO} with nonvanishing values of the coefficient~$\tau$.
Up to equivalence of reduction operators, for any regular operator~$Q$ we can set $\tau=1$.
In view of this gauge we do not need to use
the differential consequences in order to derive the determining equations,
i.e. it suffices to take into account only the equations $u_t+uu_x+u_{xx}=0$ and~$\eta-u_t-\xi u_x=0$
in the course of confining to the manifold $\mathcal{L}\cap\mathcal{Q}^{(2)}$
in the conditional invariance criterion~\eqref{EqCondInvCriterion}.
Substituting the expressions for~$u_t$ and~$u_{xx}$ obtained from these equations
into the differential function $Q_{(2)}L[u]$
and splitting the result with respect to $u_x$, we get
\begin{equation}\label{BE_RO_desystem1}
\begin{split}
&\xi_{uu}=0,\\
&-2\xi_{xu}-2\xi_u\xi+2u\xi_u+\eta_{uu}=0,\\
&2\eta_{xu}+2\xi_u\eta+\eta-\xi_t+u\xi_x-\xi_{xx}-2\xi_x\xi=0,\\
&\eta_t+u\eta_x+\eta_{xx}+2\xi_x\eta=0.
\end{split}
\end{equation}
Integrating the first two equations, we represent the functions $\xi$ and~$\eta$ as
polynomials of~$u$ with the coefficients depending on $t$ and $x$,
\begin{equation}\label{BE_RO_xieta}
\xi = \xi^1u + \xi^0,\quad
\eta=\frac13\xi^1\left(\xi^1-1\right)u^3+\left(\xi^1_x+\xi^1\xi^0\right)u^2+\eta^1u+\eta^0.
\end{equation}
Then we split the third equation of the system~\eqref{BE_RO_desystem1} with respect to~$u$
and get a system of differential equations on the functions $\xi^1$, $\xi^0$, $\eta^1$ and~$\eta^0$,
\begin{equation}\label{BE_RO_desystem2}
\begin{split}
&\xi^1(\xi^1-1)(2\xi^1+1)=0,\\
&\xi^1\xi^0(2\xi^1+1)+4\xi^1\xi^1_x=0,\\
&3\xi^1_{xx}+2(\xi^1_x\xi^0+\xi^1\xi^0_x)+(2\xi^1+1)\eta^1-\xi^1_t+\xi^0_x=0,\\
&\xi^0_t+2\xi^0\xi^0_x+\xi^0_{xx}-(2\xi^1+1)\eta^0-2\eta^1_x=0.
\end{split}
\end{equation}

The further consideration depends on the choice among the three possible solutions of the first equation
of the system~\eqref{BE_RO_desystem2}.
We rewrite the last equation of~\eqref{BE_RO_desystem1} in terms of $\xi^1$, $\xi^0$, $\eta^1$ and~$\eta^0$
and split it with respect to $u$ severally for each value of~$\xi^1$.

\subsection{Trivial case}

The case $\xi^1=1$ is rather simple.
The system~\eqref{BE_RO_desystem2} implies $\xi^0=\eta^1=\eta^0=0$.
The corresponding vector field $Q^1=\partial_t+u\partial_x$ is a unique reduction operator
(up to the equivalence relation) for the Burgers equation in this case.
The set of $Q^1$-invariant solutions consists of two families,
one of which is two-parameter and the other is one-parameter.
The two-parameter family is formed by the functions $u=(x+c_1)/(t+c_0)$, where~$c_1$ and~$c_0$ are arbitrary constants,
and all these solutions are Lie-invariant and equivalent to the scale-invariant solution $u=x/t$.
The elements of the one-parameter family are constant functions.
They are invariant with respect to shifts of both~$t$ and~$x$.

The optimal way for the construction of $Q^1$-invariant solutions of~\eqref{BE} is to integrate at first
the equation $L[u]+Q^1[u]=u_{xx}=0$, which gives the representation $u=\alpha(t)x+\beta(t)$
with some smooth functions~$\alpha$ and~$\beta$ of~$t$.
The generalized vector field $\hat Q=(L[u]+Q^1[u])\partial_u=u_{xx}\partial_u$,
is equivalent to the evolutionary representative $Q^1[u]\partial_u$ of~$Q^1$
on the set of solutions of~\eqref{BE} and hence
it is a generalized conditional symmetry of~\eqref{BE}, cf.~\cite[Proposition~4]{kunz11a}.
This is why the associated ansatz $u=\alpha(t)x+\beta(t)$ reduces the Burgers equation~\eqref{BE}
to a system of two ordinary differential equations with respect to the functions~$\alpha$ and~$\beta$,
\[
\alpha_t+\alpha^2=0,\quad
\beta_t+\alpha\beta=0.
\]

A complete set of functionally independent integrals of the equation $Q^1[u]=0$ consists of $x-ut$ and~$u$.
Therefore, directly with~$Q^1$ we construct the implicit ansatz
$u=\varphi(\omega)$, where $\omega=x-ut$ (resp.\ $x-ut=\varphi(\omega)$, where $\omega=u$)
for $Q^1$-invariant solutions with $x-ut\ne\const$ (resp.\ $u\ne\const$).
For both the ansatzes the associated reduced equations take the form $\varphi_{\omega\omega}=0$.

\subsection{Case related to Lie symmetries}

The case $\xi^1=0$
is discussed in~\cite{pucc92a} in detail (see also~\cite{arri93a}),
where it is noted that all corresponding solutions of the Burgers equation are Lie-invariant.
The system~\eqref{BE_RO_desystem2} with~$\xi^1=0$ implies
$\eta^1=-\xi^0_x$ and~$\eta^0=\xi^0_t+2\xi^0\xi^0_x+3\xi^0_{xx}$,
and splitting the last equation~\eqref{BE_RO_desystem1} with respect to~$u$ we derive
\begin{equation*}
\xi^0_{xx}=0,\qquad
\xi^0_{tt}+2\xi^0\xi^0_{tx}+4\xi^0_t\xi^0_x+4\xi^0(\xi^0_x)^2=0.
\end{equation*}
Then we have $\xi^0=\xi^{01}(t)x+\xi^{00}(t)$,
where the coefficients~$\xi^{01}$ and~$\xi^{00}$ satisfy the system
\begin{equation*}
\xi^{01}_{tt}+6\xi^{01}\xi^{01}_t+4(\xi^{01})^3=0,\qquad
\xi^{00}_{tt}+4\xi^{01}\xi^{00}_t+2\xi^{01}_t\xi^{00}+4(\xi^{01})^2\xi^{00}=0.
\end{equation*}
The transformation $\xi^{01}=\alpha_t/2\alpha$, $\xi^{00}=\beta/\alpha$ maps this system
to the system of two simple uncoupled equations $\alpha_{ttt}=0$ and~$\beta_{tt}=0$
for the functions $\alpha=\alpha(t)$ and~$\beta=\beta(t)$. Hence
\begin{equation*}
\xi^{01}=\frac{c_2t+c_1}{c_2t^2+2c_1t+c_0},\qquad
\xi^{00}=\frac{c_4t+c_3}{c_2t^2+2c_1t+c_0},
\end{equation*}
where $c_0$, \dots, $c_4$ are arbitrary constants such that $(c_0,c_1,c_2)\ne(0,0,0)$.
The substitution of the expressions obtained for $\eta^1$, $\eta^0$ and $\xi^0=\xi^{01}x+\xi^{00}$
into~\eqref{BE_RO_xieta} gives
\begin{equation*}
Q=\partial_t+\frac{(c_2t+c_1)x+c_4t+c_3}{c_2t^2+2c_1t+c_0}\partial_x
+\frac{-(c_2t+c_1)u+c_2x+c_4}{c_2t^2+2c_1t+c_0}\partial_u.
\end{equation*}
It is easy to see that the operator~$Q$ differs from a Lie symmetry operator of the Burgers equation
by the multiplier $(c_2t^2+2c_1t+c_0)^{-1}$. Thus, the following assertion is proved.

\begin{proposition}
There is a bijection between reduction operators of the Burgers equation
of the general form \[Q=\partial_t+\xi(t,x)\partial_x+\left(\eta^1(t,x)u+\eta^0(t,x)\right)\partial_u\]
and one-dimensional algebras spanned by Lie symmetry operators of this equation with nonzero coefficients of~$\p_t$.
The bijection is established by the equivalence relation of reduction operators.
\end{proposition}

In other words, any nonclassical reduction of the Burgers equation
with respect to a vector field with $\tau=1$ and $\xi_u=0$ is in fact a Lie reduction.
At the same time, any Lie solution of the Burgers equation is obtained
from a Lie solution of the linear heat equation via the Hopf--Cole transformation.
Before presenting the corresponding rigorous assertion, we recall that
the maximal Lie invariance algebra~$\mathfrak g^{\rm h}$ of the linear heat equation $v_t+v_{xx}=0$
is spanned by the vector fields
\begin{gather*}
\hat{\mathcal P}_t=\partial_t,\quad
\hat{\mathcal D}=2t\partial_t+x\partial_x,\quad
\hat{\mathcal K}=t^2\partial_t+tx\partial_x+\left(\frac14x^2-\frac12t\right)v\partial_v,\\
\hat{\mathcal P}_x=\partial_x,\quad
\hat{\mathcal G}=t\partial_x+\frac12xv\partial_v,\quad
\hat{\mathcal I}=v\partial_v,\quad
h(t,x)\partial_v,
\end{gather*}
where $h(t,x)$ runs through the set of solutions of this equation.
We associate any vector field
$Q=c_0\mathcal P_t+c_1\mathcal D+c_2\mathcal K+c_3\mathcal P_x+c_4\mathcal G$
from the maximal Lie invariance algebra~$\mathfrak g^{\rm B}$ of the equation~\eqref{BE}
with the vector field
$\hat Q=c_0\hat{\mathcal P}_t+c_1\hat{\mathcal D}+c_2\hat{\mathcal K}+c_3\hat{\mathcal P}_x+c_4\hat{\mathcal G}$
from $\mathfrak g^{\rm h}$.

\begin{proposition}
A solution~$u$ of the Burgers equation~\eqref{BE} is invariant with respect to a vector field~$Q$
from $\mathfrak g^{\rm B}$
if and only if $u=2v_x/v$ for some $\hat Q_\mu$-invariant solution~$v$ of the linear heat equation $v_t+v_{xx}=0$,
where $\mu $ is a constant and $\hat Q_\mu=\hat Q-\mu\hat{\mathcal I}\in\mathfrak g^{\rm h}$.
\end{proposition}

\begin{proof}
If smooth functions~$u$ and~$v$ depending on~$t$ and~$x$ are related via the Hopf--Cole transformation, $u=2v_x/v$, then
$Q[u]=Q[2v_x/v]=2(\hat Q[v]/v)_x.$

Suppose that $u$ is a $Q$-invariant solution of~\eqref{BE}.
Then the function~$v$ can be assumed to be a solution of the linear heat equation $v_t+v_{xx}=0$, and
$(\hat Q[v]/v)_x=Q[u]/2=0$, i.e. $\hat Q[v]=\mu v$ for some smooth function $\mu$ of~$t$.
Acting by the operator $\mathcal T=D_t+D_{xx}$ on both the sides of the last equation,
we derive
\[
\mathcal T\hat Q[v]=\hat Q[\mathcal Tv]-2(c_2t+c_1)\mathcal Tv=\mu_tv+\mu\mathcal Tv.
\]
In view of $\mathcal Tv=0$ and $\mathcal T\hat Q[v]=0$, this implies that $\mu_t=0$, i.e., $\mu$ is a constant.
As then $\hat Q[v]-\mu v=\hat Q_\mu[v]=0$, the function~$v$ is a $\hat Q_\mu$-invariant.

Conversely, if for some constant~$\mu$ the function~$v$ is
a $\hat Q_\mu$-invariant solution of the linear heat equation $v_t+v_{xx}=0$
then the function $u=2v_x/v$ is a solution of the Burgers equation~\eqref{BE} and
$Q[u]=Q[2v_x/v]=2(\hat Q[v]/v)_x=2(\hat Q_\mu[v]/v)_x=0$, i.e.,
the function~$u$ is $Q$-invariant.
\end{proof}

\subsection{No-go case}

The case $\xi^1=-\frac12$
leads to reduction operators of the general form
\begin{equation}\label{QforCase12}
Q=\partial_t+\left(-\frac{1}{2}u+\xi^0\right)\partial_x
+\left(\frac14u^3-\frac{\xi^0}{2}u^2+\eta^1u+\eta^0\right)\partial_u,
\end{equation}
where the coefficients $\xi^0$, $\eta^1$ and~$\eta^0$ are smooth functions of~$t$ and~$x$ satisfying
the system of differential equations
\begin{gather}\label{BE_RO_desystem_case12}
\begin{split}
&\xi^0_t+2\xi^0\xi^0_x+\xi^0_{xx}-2\eta^1_x=0,\\
&\eta^1_t+2\xi^0_x\eta^1+\eta^1_{xx}+\eta^0_x=0,\\
&\eta^0_t+2\xi^0_x\eta^0+\eta^0_{xx}=0
\end{split}
\end{gather}
derived from~\eqref{BE_RO_desystem2}.
As a differential substitution reduces the system~\eqref{BE_RO_desystem_case12}
to an uncoupled system of three copies of the linear heat equation~\cite{arri02a,mans99a},
the general solution of~\eqref{BE_RO_desystem_case12} cannot be represented in a closed form.
Hence the case~$\xi^1=-\frac12$ is called a ``no-go'' case.
This result appears to directly follow from the fact that $Q$ is a reduction operator of the equation~\eqref{BE}.
Moreover, we show that solutions of the system~\eqref{BE_RO_desystem_case12} can be represented
via solutions of the uncoupled system of three copies of the Burgers equation.

\begin{theorem}\label{TheoremOnNogoCaseofRedPosOfBurgersEq}
Any solution of the determining system~\eqref{BE_RO_desystem_case12}
on the coefficients of reduction operators of the form~\eqref{QforCase12} is represented as
\begin{equation}\label{BE_RO_case12_coeffs_via_phi}
\xi^0=\frac{(W(\bar v))_x}{W(\bar v)},
\quad
\eta^1=\displaystyle\frac{|\bar v,\bar v_{xx},\bar v_{xxx}|}{W(\bar v)},
\quad
\eta^0=-2\displaystyle\frac{W(\bar v_x)}{W(\bar v)},
\end{equation}
where $\bar v=(v^1,v^2,v^3)$ is a triple of linear independent solutions of the heat equation $v_t+v_{xx}=0$,
$W(\bar v)=|\bar v,\bar v_x,\bar v_{xx}|$ and~$W(\bar v_x)=|\bar v_x,\bar v_{xx},\bar v_{xxx}|$
are the Wronskians of this triple
and the triple of the corresponding derivatives with respect to~$x$, respectively,
and $|\bar p,\bar q,\bar r|$ denotes
the determinant of the matrix constructed from ternary columns~$\bar p$, $\bar q$ and~$\bar r$.
Conversely, any triple~$(\xi^0,\eta^1,\eta^0)$ admitting the representation~\eqref{BE_RO_case12_coeffs_via_phi}
satisfies the system~\eqref{BE_RO_desystem_case12}.
\end{theorem}

\begin{proof}
We fix an operator~$Q$ of the form~\eqref{QforCase12}.
The set of $Q$-invariant solutions of the Burgers equation $L[u]=0$
coincides with the set of solutions of the system $L[u]=0$ and $Q[u]=0$,
and it is parameterized by two arbitrary constants
as $Q$ is a regular reduction operator of the Burgers equation~\cite{kunz08b}.
For convenience we recombine the equations of the above system in the following way: $L[u]=0$, $L[u]+Q[u]=0$.
The Hopf--Cole transformation $u=2v_x/v$ maps this system to the linear system
\begin{equation}\label{BE_case12_3OLDE}
v_t+v_{xx}=0, \qquad
v_{xxx}-\xi^0v_{xx}+\eta^1v_x+\frac12\eta^0v=0.
\end{equation}
Let for some integer~$n$ functions $v^1$,~\dots,~$v^n$ of~$t$ and~$x$ be
linear independent solutions of the system~\eqref{BE_case12_3OLDE}.
Then the equation
\begin{equation}\label{EqQinvSolutions12OfBE}
u=2\frac{c_1v^1_x+\dots+c_nv^n_x}{c_1v^1+\dots+c_nv^n}
\end{equation}
where $c_1$,~\dots,~$c_n$ are arbitrary constants which are not simultaneously zero,
defines a family of $Q$-invariant solutions of the Burgers equation parameterized by $n-1$ essential constants,
and hence $n\leqslant3$ as the number of such parameters cannot be greater than two.
In other words, the dimension of the space~$V$ of solutions of the system~\eqref{BE_case12_3OLDE}
does not exceed three.
This dimension cannot also be less than three.
Indeed, let now the functions $v^1$,~\dots,~$v^n$ of~$t$ and~$x$ form a basis of the space~$V$, where $n=\dim V$.
Then the expression~\eqref{EqQinvSolutions12OfBE} represents the general solution of the system $L[u]=0$ and $Q[u]=0$,
which contains $n-1$ essential constant parameters. Therefore, $n-1=2$, i.e.~$n=3$.

Consider a basis $\{v^1,v^2,v^3\}$ of the space~$V$.
By definition, the elements of~$V$ are solutions of the heat equation $v_t+v_{xx}=0$, which is linear and evolutionary.
Hence the usual linear independence of them implies the linear independence of them
over the ring of smooth functions of~$t$,
i.e.\ the Wronskian~$W(\bar v)$ of the functions~$v^1$, $v^2$ and $v^3$ with respect to~$x$ does not vanish.
See, e.g., Note~5 in~\cite{popo08b}.
Substituting the elements of the basis into the second equation of~\eqref{BE_case12_3OLDE},
we obtain a well-defined system of linear algebraic equations,
\begin{equation*}
v^i_{xxx}-\xi^0v^i_{xx}+\eta^1v^i_x+\frac12\eta^0v^i=0, \quad i=1,2,3,
\end{equation*}
for the coefficients~$\xi^0$, $\eta^1$, $\eta^0$,
or, in the matrix form,
$M\bar{q}=\bar v_{xxx}$,
where
\begin{equation*}
\bar v=
\left(
\begin{array}{c}
v^1\\ v^2\\ v^3
\end{array}
\right),
\quad
M=
\left(
\begin{array}{ccc}
v^1&v^1_x&v^1_{xx}\\
v^2&v^2_x&v^2_{xx}\\
v^3&v^3_x&v^3_{xx}
\end{array}
\right),
\quad
\bar{q}=\left(
\begin{array}{c}
-\frac12\eta^0\\ -\eta^1\\ \xi^0
\end{array}
\right).
\end{equation*}
Solving this system with respect to~$\xi^0$, $\eta^1$ and~$\eta^0$,
we derive the representation~\eqref{BE_RO_case12_coeffs_via_phi}.

As the proof can be turned back, the inverse statement is also true.
\end{proof}

\begin{corollary}
The coefficients of the reduction operator~\eqref{QforCase12} of the Burgers equation
can be represented in the form
\begin{gather}\label{BE_RO_case12_coeffs}
\xi^0=\frac{1}{2}
\frac{|\bar e,\bar u,\bar z|}
{|\bar e,\bar u,\bar y|},
\quad
\eta^1=\frac{1}{4}
\frac{|\bar e,\bar y,\bar z|}
{|\bar e,\bar u,\bar y|},
\quad
\eta^0=-\frac{1}{4}
\frac{|\bar u,\bar y,\bar z|}
{|\bar e,\bar u,\bar y|},
\end{gather}
where the columns~$\bar e $, $\bar y$ and~$\bar z$ consist of
three units, the expressions $y^i=2u^i_x+(u^i)^2$ and~$z^i=4u^i_{xx}+6u^iu^i_x+(u^i)^3$, respectively, $i=1,2,3$,
and $\bar u$ is a column of three solutions of the Burgers equation with $|\bar e,\bar u,\bar y|\ne0$.
\end{corollary}

\begin{proof}
The connection between solutions of the heat equation and the Burgers equation
via the Hopf--Cole transformation $2v^i_x/v^i=u^i$ gives the expressions
\begin{equation*}
\displaystyle\frac{v^i_{xx}}{v^i}=\frac12u^i_x+\frac14(u^i)^2,
\quad
\displaystyle\frac{v^i_{xxx}}{v^i}=\frac34u^iu^i_x+\frac18(u^i)^3+\frac12u^i_{xx}, \quad i=1,2,3.
\end{equation*}
The substitution of these expressions into~\eqref{BE_RO_case12_coeffs_via_phi} proves the corollary.
Note that the determinant $|\bar e,\bar u,\bar y|$ is nonvanishing
as the Wronskian~$W(\bar v)$ is the same.
\end{proof}

\begin{corollary}\label{CorollaryOnBijectionBetweenSolutionsAndRedOpsInNogoCaseForBurgersEq}
The representations~\eqref{BE_RO_case12_coeffs_via_phi} and~\eqref{EqQinvSolutions12OfBE}, where $n=3$,
explicitly define the one-to-one correspondence
between reduction operators of the form~\eqref{QforCase12} and families of solutions of the Burgers equation that
are invariant with respect to these operators.
\end{corollary}

It looks very difficult to explicitly construct an ansatz for~$u$
by the direct integration of the invariant surface condition
that corresponds to an operator of the form~\eqref{QforCase12}
in the case of an arbitrary solution of system~\eqref{BE_RO_desystem_case12}.
Even for a simple solution of system~\eqref{BE_RO_desystem_case12},
carrying out the corresponding reduction of the Burgers equation~\eqref{BE} is not a trivial problem.
See, e.g., \cite{arri93a,pucc92a}, where a few such reductions were implemented.
A better way for using an operator of the form~\eqref{QforCase12} for reducing the Burgers equation
is to consider, instead of~$Q$, the generalized vector field $\hat Q=(L[u]+Q[u])\partial_u$
which is equivalent to the evolutionary representative $Q[u]\partial_u$ of~$Q$ and
is a generalized conditional symmetry of~\eqref{BE},
cf.~\cite[Proposition~4]{kunz11a}.
In fact, the reduction of equation~\eqref{BE} with respect to the generalized vector field $\hat Q$
is somehow presented in the proof of Theorem~\ref{TheoremOnNogoCaseofRedPosOfBurgersEq},
and Corollary~\ref{CorollaryOnBijectionBetweenSolutionsAndRedOpsInNogoCaseForBurgersEq}
exhaustively describes the family of $Q$-invariant solutions of~\eqref{BE}.

At the same time, knowing the representation~\eqref{EqQinvSolutions12OfBE}, where $n=3$,
for $Q$-invariant solutions of~\eqref{BE} allows us to easily construct an ansatz for~$u$
associated with the operator~$Q$ and then reduce the Burgers equation~\eqref{BE}
using this ansatz.
Setting $c_1=1$ and $c_2=0$ (resp.\ $c_1=0$ and $c_2=1$) and assuming that $c_3$ is an arbitrary constant,
we derive two integrals of the equation $Q[u]=0$,
\[
\zeta =\frac{v^1u-2v^1_x}{v^3u-2v^3_x},\quad
\omega=\frac{v^2u-2v^2_x}{v^3u-2v^3_x},
\]
which are correctly defined for $u\ne2v^3_x/v^3$
(we can always assume this condition to be satisfied for a specific solution~$u$
up to renumbering the functions~$v^1$, $v^2$ and $v^3$).
Therefore, the general solution of the equation $Q[u]=0$ is implicitly represented in the form $F(\zeta,\omega)=0$,
where $F$ is an arbitrary nonconstant function of its arguments.
Up to the permutation of~$\zeta$ and~$\omega$, the derivative $F_\zeta$ can be assumed nonvanishing.
Then the equality $F(\zeta,\omega)=0$ implies the implicit ansatz \[\zeta=\varphi(\omega)\]
for the unknown function $u=u(t,x)$,
where $\zeta$ and~$\omega$ are the new dependent and independent variables, respectively.
The standard way for deriving the corresponding reduced equation via
the computation of the expressions for derivatives $u_t$, $u_x$ and $u_{xx}$ implied by the ansatz
and the subsequent substitution of the expressions into~\eqref{BE} is too cumbersome.
Instead of this, we act on the ansatz $\zeta=\varphi(\omega)$ by the operator $D_t+D_{xx}$
and then make the substitutions $u_t+u_{xx}=-uu_x$ and $v^i_t+v^i_{xx}=0$, $i=1,2,3$.
As a result, we derive the equation $\varphi_{\omega\omega}D_x\omega=0$.

The differential function~$D_x\omega$ does not vanish for solutions of~\eqref{BE}
which are implicitly represented in the form $\zeta=\varphi(\omega)$ for some smooth function~$\varphi$ of~$\omega$.
Indeed, suppose that $D_x\omega=0$ (i.e., $\omega=\omega^0(t)$) for such a solution $u=u^0(t,x)$.
Then we have $u^0=2(v^2_x-\omega^0v^3_x)/(v^2-\omega^0v^3)$,
which implies that for some nonvanishing smooth function~$\chi$ of~$t$
the expression $(v^2+\omega^0v^3)\chi$ gives a solution of the linear heat equation $v_t+v_{xx}=0$.
As the functions~$v^2$ and~$v^3$ are linearly independent solutions of the same equation,
this is possible only if $\omega^0$ is a constant.
Therefore, the substitution $u=u^0(t,x)$ into~$\zeta$ results in the constant value $\zeta^0=\varphi(\omega^0)$
and hence $u^0=2(v^1_x-\zeta^0v^3_x)/(v^1-\zeta^0v^3)$. Comparing the two expressions for the solution $u=u^0(t,x)$,
we conclude that the functions~$v^1$, $v^2$ and~$v^3$ are linearly dependent and thus arrive at a contradiction.

As $D_x\omega\ne0$, the reduced equation takes the form \[\varphi_{\omega\omega}=0.\]
Its general solution is $\varphi=\tilde c_1\omega+\tilde c_2$,
which completely agrees with the representation~\eqref{EqQinvSolutions12OfBE} and the ansatz $\zeta=\varphi(\omega)$.
The above reduced equation is obtained for an arbitrary operator~$Q$ of the form~\eqref{QforCase12} and
it is much simpler than particular reduced equations derived in~\cite{arri93a,pucc92a}.
This fact is explained by that the ansatz $\zeta=\varphi(\omega)$,
whose construction is based on the representation~\eqref{EqQinvSolutions12OfBE}
for $Q$-invariant solutions of the Burgers equation~\eqref{BE},
is in better agreement with the structure of this equation than particular ansatzes given in~\cite{arri93a,pucc92a}.

\section{Conclusion}

The aim of this paper is to arrange, enhance and complete the description of the nonclassical reductions
of the Burgers equation, including Lie reductions.
Although this problem had been considered in a number of papers, certain of its aspects needed additional investigation.

The set of reduction operators of the Burgers equation is naturally partitioned into two subsets,
which consist of singular and regular reduction operators, respectively.
Basic properties of singular reduction operators, whose coefficients of~$\partial_t$ vanish,
are quite common for (1+1)-dimensional evolution equations and properly formulated in no-go terms,
and the main property is the existence of a bijection
between equivalence classes of singular reduction operators and one-parameter families of solutions
which are different up to re-parameterization.
At the same time, the subset of regular reduction operators,
whose coefficients of~$\partial_t$ do not vanish, is of specific structure.
The representatives of equivalence classes of such operators, whose coefficients of~$\partial_t$ equal one,
are partitioned into three sets, namely, the unique operator~$\partial_t+u\partial_x$ with~$\xi_u=1$,
the family of operators with~$\xi_u=0$, each of which is equivalent to a Lie symmetry operator, and
the family of operators with~$\xi_u=-\frac12$, which have the form~\eqref{QforCase12}
with the coefficients~$\xi^0$, $\eta^1$ and~$\eta^0$ satisfying the system~\eqref{BE_RO_desystem_case12}.

We present a new optimized proof of the no-go theorem on the system~\eqref{BE_RO_desystem_case12} of determining equations
for reduction operators with $\tau=1$ and $\xi_u=-\tfrac12$, which is directly based on properties of such operators.
As a consequence of the theorem, it is detected
that the coefficients of reduction operators in this case
admit the representation in terms of solutions of the uncoupled system formed by three copies of the Burgers equation.
Any Lie reduction of the Burgers equation proves to be equivalent via the Hopf--Cole transformation
to a parameterized family of Lie reductions of the linear heat equation.
We also carry out all possible nonclassical reductions of the Burgers equation to single ordinary differential equations.
Regular reduction operators with the coefficient~$\tau$ gauged to one give
purely nonclassical reductions of the Burgers equation only if $\xi_u\in\{1,-\tfrac12\}$.
It is obvious that all corresponding ansatzes are necessarily implicit and hence they cannot be constructed
using the direct method by Clarkson and Kruskal.
See related discussions on the connection between nonclassical symmetries and the direct method, e.g.,
in~\cite{clar89a,pucc92a}.

Finding reduction operators for the Burgers equation, one faces two kinds of no-go cases.

The first kind is given by singular reduction operators and is in fact related
to lowering the equation order to one in the course of reduction.
This is why similar no-go results are true for any (1+1)-dimensional partial differential equation
possessing a family of reduction operators of singularity co-order one
which is parameterized by an arbitrary smooth function of all independent and dependent variables~\cite{kunz08b}.
For each such family, the system of determining equations consists of
a single partial differential equation on the function parameterizing the family,
and this equation is equivalent, in a certain sense, to the original equation,
where the variable tuple is implicitly augmented with an additional parametric variable.
The above results can even be extended to multidimensional partial differential equations~\cite{boyk12a}.

The second kind given by regular reduction operators with $\tau=1$ and $\xi_u=-\tfrac12$ is more specific
and definitely related to the fact that the Burgers equation is linearized by the Hopf--Cole transformation
to the linear heat equation.
The corresponding system of determining equations is a system of three (1+1)-dimensional evolution equations,
which is reduced by differential substitutions to the uncoupled system of three copies of the linear heat equation
as well as to the uncoupled system of three copies of the Burgers equation.
Similar no-go results are known only for linear (1+1)-dimensional evolution equations
of the second order~\cite{fush92a,fush94a,popo95a,popo08b}.
It looks possible to extend these results to the entire class of generalized Burgers equations
which are linearized by the Hopf--Cole transformation,
and the optimized proof of Theorem~\ref{TheoremOnNogoCaseofRedPosOfBurgersEq} creates a significant prerequisite for this.
The question whether there exist no-go cases of other kinds
related to regular reduction operators of single evolution equations is still open.

\subsection*{Acknowledgments}

The authors thank Vyacheslav Boyko for useful discussions.
The research of ROP was supported by the Austrian Science Fund (FWF), projects P20632 and Y237.

\end{document}